\documentclass[runningheads]{llncs}
\usepackage[T1]{fontenc}
\usepackage{amssymb}
\usepackage{algorithm}
\usepackage{mathtools}
\usepackage{hhline}
\usepackage{algpseudocode}
\usepackage{xcolor}
\usepackage{makecell}
\usepackage{stmaryrd}
\usepackage{hyperref}
 
\usepackage{cleveref}

\usepackage{graphicx}
\begin{document}
\title{Effective alpha theory certification using interval arithmetic: alpha theory over regions}
\titlerunning{Alpha theory over regions}
% If the paper title is too long for the running head, you can set
% an abbreviated paper title here
%
\author{Kisun Lee\inst{1}\orcidID{0000-0003-1191-1400}}
%
% First names are abbreviated in the running head.
% If there are more than two authors, 'et al.' is used.
%
\institute{Clemson University, Clemson SC 29634, USA 
\email{kisunl@clemson.edu}\\
\url{https://klee669.github.io}}
\maketitle              % typeset the header of the contribution
\begin{abstract}
We reexamine Smale's alpha theory as a way to certify a numerical solution to an analytic system. For a given point and a system, Smale's alpha theory determines whether Newton's method applied to this point shows the quadratic convergence to an exact solution. We introduce the alpha theory computation using interval arithmetic to avoid costly exact arithmetic. As a straightforward variation of the alpha theory, our work improves computational efficiency compared to software employing the traditional alpha theory.

\end{abstract}
\section{Introduction}

The primary focus of this paper is to \emph{certify} a numerical solution to an analytic system $F:U\rightarrow\mathbb{C}^n$ defined in an open set $U\subset \mathbb{C}^n$. Certifying a solution is determining if a given point can be refined to an exact solution up to arbitrary precision. %Notably, \emph{Smale's alpha theory} \cite[Chapter 8]{blum2012complexity} addresses certification in the context of utilizing Newton's method for refinement. This paper aims to perform the task of the alpha theory using interval arithmetic for efficient computation.
\emph{Newton's method} is a widely recognized method for this task. For an analytic system $F$, we define the \emph{Newton operator} $N_F(x)$ for $F$ by
\[N_F(x)=\left\{\begin{array}{ll}
    x-JF(x)^{-1}F(x) & \text{if }JF(x)\text{ is invertible,}\\
    x & \text{otherwise.}
\end{array}\right.\]
This operator is applied iteratively at a point $x\in \mathbb{C}^n$ to approximate the exact solution $x^\star$ of $F$. Especially, when the point $x$ is ``close'' to $x^\star$,  it is known that $x$ shows the quadratic convergence to $x^\star$. Extensive research on the convergence of Newton's method has led to the development of \emph{Smale's alpha theory} \cite[Chapter 8]{blum2012complexity}, which establishes criteria for a point and a system to show the quadratic convergence of Newton's method.

A drawback of employing the alpha theory for certification lies in the necessity for computationally expensive exact arithmetic to ensure its rigor. Although the Krawczyk method \cite{breiding2023certifying,burr2019effective,krawczyk1969newton,lee2019certifying,moore2009introduction} utilizes faster arithmetic, it does not always guarantee the quadratic convergence.% for a numerical solution and does not demonstrate potential for certifying solutions with multiplicity greater than $1$ (referred to as a \emph{multiple zero}).

This paper introduces the \emph{alpha theory over regions}. It executes the alpha theory computation with an interval (vector) input containing a candidate solution to enable efficient interval arithmetic. By developing the alpha theory over regions, our work introduces novel aspects; the enhancement of certification speed for numerical solutions, surpassing the method outlined in \cite{hauenstein2012algorithm}.%, but also the proposition of a capability for certifying solutions with multiple zeros.

The structure of the paper is as follows. In Section \ref{sec:prelim}, we review the concepts of Smale's alpha theory and interval arithmetic. In Section \ref{sec:alpha-regions}, the alpha theory over regions is introduced, which is the main result of the paper. %Section \ref{sec:double-zero} proposes the alpha theory over regions with an isolated solution with multiplicity $2$ to highlight its potential to certify a multiple zero. 
Some remarks for implementing the alpha theory over regions are discussed in Section \ref{sec:implementation-details}. 
The experimental results are provided in Section \ref{sec:experiments}.
 
\section{Preliminaries}\label{sec:prelim}

We introduce the concepts needed for the alpha theory over regions. Initially, we discuss the alpha theory, followed by a review of interval arithmetic.

\subsection{Smale's alpha theory}\label{sec:alpha_theory}

%In this section, we recall Smale's alpha theory as a method for checking the quadratic convergence of an approximation of a solution to a square system of equations. 
Let $F:U\rightarrow \mathbb{C}^n$ be an analytic system defined in an open set $U\subset\mathbb{C}^n$.
Then, for a point $x\in\mathbb{C}^n$, the \emph{$k$-th Newton iteration} $N_F^k(x)$ is defined by applying the Newton operator $k$ times at $x$. For an exact solution $x^\star$ to the system $F$, suppose that we have
$\left\|N_F^k(x)-x^\star\right\|\leq \left(\frac{1}{2}\right)^{2^k-1}\|x-x^\star\|$
for every $k\in\mathbb{N}$. Then, we say that $x$ \emph{converges quadratically} to $x^\star$, and $x$ is called an \emph{approximate solution} to $F$ with the \emph{associated solution} $x^\star$. In other words, certifying $x$ means proving $x$ is an approximate solution to $F$ associated with $x^\star$. %In short, the alpha theory determines whether a point is an approximate solution to the system $F$.

The alpha theory exploits three values obtained from $x$ and $F$. If the Jacobian $JF(x)$ is invertible, we define  
\begin{equation*}
	\begin{array}{ll}
	\alpha(F,x) :=  \beta(F,x)\gamma(F,x)&\gamma(F,x)  :=  \sup\limits_{k\geq 2}\left\|\frac{JF(x)^{-1}J^kF(x)}{k!}\right\|^{\frac{1}{k-1}}\\
	\beta(F,x)  :=  \|x-N_F(x)\|=\|JF(x)^{-1}F(x)\|&
	\end{array}
\end{equation*}
where $J^kF(x)$ is a symmetric tensor whose components are $k$-th order partial derivative of $F$. The value $\beta(F,x)$ is the Euclidean norm of Newton step for $F$ at $x$. The norm used in $\gamma(F,x)$ is the operator norm for $JF(x)^{-1}J^kF(x)$ which is induced from the norm on the $k$-fold symmetric power $S^k\mathbb{C}^n$ of $\mathbb{C}^n$. When $JF(x)^{-1}$ is not invertible, we define $\alpha(F,x)=\beta(F,x)=\gamma(F,x)=\infty$. Results of the alpha theory are like the following:

\begin{theorem}\cite[Section 8.2, Theorem 2, Theorem 4 and Remark 6]{blum2012complexity}\label{thm:alpha-1} Let $F:U\rightarrow \mathbb{C}^n$ be an analytic system, and $x$ be a given point in $U$. Then,
\begin{enumerate}
    \item if 
$\alpha(F,x)<\frac{13-3\sqrt{17}}{4}$,
then $x$ is an approximate solution to $F$. Moreover, $\|x-x^\star\|\leq 2\beta(F,x)$ where $x^\star$ is the associated solution to $x$, and 
\item \label{thm:distinct} if $\alpha(F,x)<0.03$ and $\|x-y\|<\frac{1}{20\gamma(F,x)}$ for some $y\in U$, then $x$ and $y$ are approximate solutions to the same associated solution to $F$. In addition, there is a unique solution $x^\star$ to $F$ in the ball $B(x,\frac{1}{20\gamma(F,x)})$ centered at $x$ with the radius $\frac{1}{20\gamma(F,x)}$. 
\end{enumerate}
\end{theorem}

The first part of Smale's alpha theory checks if a point $x$ is an approximate solution to a system $F$.
The second part of the alpha theory identifies when two different points have the same associated solution to $F$, that is, it certifies distinct numerical solutions.

\iffalse
As a corollary, the theorem above can be used to certify the realness of solutions to a polynomial system. We say that a polynomial system $F=\{f_1,\dots,f_n\}$ is \emph{real} if $\{\overline{f}_1,\dots, \overline{f}_n\}=\{f_1,\dots, f_n\}$ where $\overline{f}$ is a complex conjugate of $f$. Solutions to a real polynomial system appear as real or conjugate complex pairs. Using this fact, the corollary below establishes the conditions for a point that approximates a real solution to the real polynomial system.

\begin{corollary}\cite[Section 2.1]{hauenstein2012algorithm}\label{cor:alpha-real}
Let $F:\mathbb{C}^n\rightarrow\mathbb{C}^n$ be a real polynomial system and let $x$ be a point in $\mathbb{C}^n$. If
\[\alpha(F,x)<0.03\]
then an exact solution $x^\star$ to $F$ associated to $x$ is real if $\|x-\overline{x}\|\leq \frac{10}{3}\beta(F,x)$, and not real if $\|x-\overline{x}\|>4\beta(F,x)$ where $\overline{x}$ is the complex conjugate of $x$. 
\end{corollary}
\fi

The most challenging part of implementing the alpha theory is the computation of the gamma value. To resolve this issue, a known approach is to find an upper limit for the gamma value. For a polynomial $f=\sum\limits_{|\nu|\leq d}a_\nu x^\nu$ of degree $d$, we define the \emph{Bombieri-Weyl norm} $\|f\|^2=\frac{1}{d!}\sum\limits_{|\nu|\leq d}\nu!(d-|\nu|)!|a_\nu|^2$. This norm extends to a polynomial system $F=\{f_1,\dots, f_n\}$ with $\|F\|^2=\sum\limits_{i=1}^n\|f_i\|^2$. For a point $x\in \mathbb{C}^n$, we define $\|(1,x)\|^2=1+\sum\limits_{i=1}^n |x_i|^2$. Let $d_i$ be the degree of each polynomial $f_i$, and set $D=\max\limits_i\{d_i\}$. Finally, define the diagonal matrix $\Delta_F(x)$ whose diagonal entry is given by $\Delta_F(x)_{ii}:=\sqrt{d_i} \|(1,x)\|^{d_i-1}$. Then, the following result provides an upper bound for $\gamma(F,x)$ for a polynomial system $F$:
\begin{proposition}\cite[Section I-3, Proposition 3]{shub1993complexity}\label{prop:higher-bound}
    Consider a polynomial system $F:\mathbb{C}^n\rightarrow\mathbb{C}^n$ and a point $x\in\mathbb{C}^n$ such that $JF(x)$ is invertible. We define 
    $\mu(F,x):=\max\{1,\|F\|\|JF(x)^{-1}\Delta_F(x)\|\}$
    with the operator norm for $JF(x)^{-1}\Delta_F(x)$. Then, 
    $\gamma(F,x)\leq \frac{\mu(F,x)D^{\frac{3}{2}}}{2\|(1,x)\|}$.
\end{proposition}

Upper bounds for the gamma value have been studied in various instances of systems of equations. In the case of systems involving polynomial-exponential equations, \cite{hauenstein2017certifying} accomplished this. Additionally, for a broader range of systems, \cite{burr2019effective} introduced an upper bound for the gamma value when dealing with systems with univariate $D$-finite functions.

There are known implementations of the alpha theory. For standalone software, \texttt{alphaCertified} \cite{hauenstein2011alphacertified} is used. On the other hand, the \texttt{Macaulay2} package \texttt{NumericalCertification} \cite{lee2019certifying} provides a specialized implementation for computation in \texttt{Macaulay2} \cite{M2}.

\subsection{Interval arithmetic}

%As a method to provide an input region for certification, we use an interval (vector) that is contained in $\mathbb{C}^n$. We introduce \emph{interval arithmetic} to handle the interval used as an input.

Interval arithmetic introduces operations between intervals to perform conservative computations to produce a certified result. For example, for two intervals $[a,b]$ and $[c,d]$ over $\mathbb{R}$, and an arithmetic operation $\odot$, we define 
$[a,b]\odot[c,d]:=\{x\odot y\mid x\in [a,b], y\in [c,d]\}$.
For explicit formulas for the standard arithmetic operations (e.g.\ $+,-,\cdot,/$), see \cite{moore2009introduction}.
Interval arithmetic can be extended to $\mathbb{C}$ by introducing an interval with real and imaginary parts, that is, $I=\Re(I)+i\Im(I)$.
\iffalse
For two intervals $I_1=\Re(I_1)+i\Im(I_1)$ and $I_2=\Re(I_2)+i\Im(I_2)$, assuming interval arithmetic over $\mathbb{R}$, we have
\begin{align*}
I_1+I_2&=(\Re(I_1)+\Re(I_2))+i(\Im(I_1)+\Im(I_2))\\
I_1-I_2&=(\Re(I_1)-\Re(I_2))+i(\Im(I_1)-\Im(I_2))\\
I_1\cdot I_2&=(\Re(I_1)\cdot\Re(I_2)-\Im(I_1)\cdot \Im(I_2))\\
&\quad\quad\quad\quad\quad+i(\Re(I_1)\cdot\Im(I_2)+\Im(I_1)\cdot\Re(I_2))\\
I_1/I_2 &= \frac{\Re(I_1)\cdot\Re(I_2)+\Im(I_1)\cdot\Im(I_2)}{\Re(I_2)\cdot \Re(I_2)+\Im(I_2)\cdot \Im(I_2)}\\
&\quad\quad\quad\quad\quad+i\frac{\Im(I_1)\cdot\Re(I_2)-\Re(I_1)\cdot\Im(I_2)}{\Re(I_2)\cdot \Re(I_2)+\Im(I_2)\cdot \Im(I_2)}\quad\text{if }0\not\in I_2.
\end{align*}
\fi

%We provide further notations useful for the rest of the paper. 
For an interval $I$ in $\mathbb{R}$, 
%the maximum absolute value at the points in $I$ is defined by $\lceil I\rceil=\max\limits_{x\in I}|x|$. Likewise,
we define the minimum absolute value over the points in $I$ by $\lfloor I\rfloor=\min\limits_{x\in I}|x|$. 
%Let $m(I)$ denotes the midpoint of $I$. 
Consider an interval vector $I=(I_1,\dots, I_n)$ in $\mathbb{R}^n$. We define 
%$\llceil I\rrceil^2=\sum\limits_i\lceil I_i\rceil^2$ and 
$\llfloor I\rrfloor^2=\sum\limits_i \lfloor I_i\rfloor^2$. Note that $\llfloor I\rrfloor$ is the minimum of $2$-norms over all points in $I$, and it can be extended to intervals in $\mathbb{C}$ naturally. 
%For instance, for an interval $I$ in $\mathbb{C}$, we have $\lceil I\rceil=\max\limits_{x\in I}|x|$.
%Moreover, for an interval vector $I=(I_1,\dots, I_n)$ in $\mathbb{C}^n$, we define the real and imaginary parts of $I$ by $\Re(I)=(\Re(I_1),\dots, \Re(I_n))$ and $\Im(I)=(\Im(I_1),\dots, \Im(I_n))$ respectively.
%As interval arithmetic replaces arithmetic over numbers by intervals, we need to consider an interval-valued function, especially that is derived from a function over a usual number field. 
For a function $F:\mathbb{C}^n\rightarrow\mathbb{C}^m$ and an interval vector $I$ in $\mathbb{C}^n$, we define an \emph{interval closure} $\square F(I)$ of $F$ over $I$ which is a set containing $\{F(x)\mid x\in I\}$. Usually, the smallest interval in $\mathbb{C}^m$ that contains $\{F(x)\mid x\in I\}$ is used for $\square F(I)$. For a function $F$ that consists of elementary functions (e.g.\ polynomials), the interval closure is obtained by interval arithmetic. Note that the interval closure of a function is not unique in general.

Lastly, we consider an \emph{interval matrix}, a matrix with interval entries. For an $n\times n$-interval matrix $M$, we say that an $n\times n$-interval matrix is an \emph{inverse interval matrix} of $M$ if it contains the set $\{N^{-1}\mid N\in M\}$, which is denoted by $M^{-1}$. Note that $M^{-1}$ may not be unique. We discuss how to compute the inverse interval matrix in Section \ref{sec:inverse-interval-matrix}.

%Lastly, we consider an \emph{interval matrix}, a matrix with interval entries. For an $m\times n$-interval matrix $M$ in $\mathbb{C}^{m\times n}$, we define  $\llceil M\rrceil=\max\limits_i\sum\limits_{j}\lceil M_{ij}\rceil$ and $\llfloor M\rrfloor=\max\limits_i\sum\limits_{j}\lfloor M_{ij}\rfloor$ where $M_{ij}$ is $(i,j)$-entry of $M$. 
%Note that $\llceil M\rrceil$ (or $\llfloor M\rrfloor$, respectively) imply the maximum (or minimum, respectively) operator norm of a matrix in $M$ under the max-norm. 
%Furthermore, for an $n\times n$-interval matrix $M$, we say that $M^{-1}$ is an \emph{inverse interval matrix} of $M$ if it contains the set $\{N^{-1}\mid N\in M\}$. Note that $M^{-1}$ may not be unique. We discuss how to compute the inverse interval matrix in Section \ref{sec:inverse-interval-matrix}.

\section{The alpha theory over regions}\label{sec:alpha-regions}
The goal of this section is to extend the results from Section \ref{sec:alpha_theory} to the case when the input is given by an interval vector rather than a point. Let $F:U\rightarrow \mathbb{C}^n$ be an analytic system defined in an open set $U\subset \mathbb{C}^n$, and $I$ be an interval vector that is contained in $U$. Then, we define the three values given by $F$ and $I$ like the following:
\begin{equation*}
	\begin{array}{l}
	\alpha(F,I) :=  \beta(F,I)\gamma(F,I)\\
	\beta(F,I)  :=  \max\limits_{x\in I}\beta(F,x)\\
	\gamma(F,I)  :=  \max\limits_{x\in I}\gamma(F,x)
	\end{array}.
\end{equation*}
If $JF(x)$ is not invertible at some point $x\in I$, we define $\alpha(F,I)=\beta(F,I)=\gamma(F,I)=\infty$. %Otherwise, we note that $\beta(F,I)$ and $\gamma(F,I)$ are well-defined since $F$ is analytic on $U$ and $I$ is a compact set. 
%The way to compute these values will be mentioned in 

%Recall that alpha theory not only provides a certificate but also returns the region that an exact solution must lie in. Alpha theory over regions also returns such region but it is not a ball nor interval box. For an interval box $I$ and a positive constant $r$, we define a \emph{rounded rectangle} 

We state the results of the alpha theory over regions. 
\begin{theorem}\label{thm:alpha-regions}
    Let $F:U\rightarrow \mathbb{C}^n$ be an analytic system, and $I$ be an interval vector in $U$. Then,
    \begin{enumerate}
        \item \label{thm:alpha-regions1} if $\alpha(F,I)<\frac{13-3\sqrt{17}}{4}$,
    then all points in $I$ are approximate solutions to $F$ with the associated solution $x^\star$ that is contained in $\bigcap\limits_{x\in I}B(x,2\beta(F,I))$, and %In addition, if we let $d$ be the diagonal of $I$ (i.e.\  the longest distance between two points in $I$), then $\beta(F,I)\geq \frac{d}{4}$, and
    \item \label{thm:distinct-regions} if $\alpha(F,I)<0.03$,
    then all points in  
    $\bigcup\limits_{x\in I}B(x,\frac{1}{20\gamma(F,I)})$ are approximate solutions to the same associated solution to $F$.    
    \end{enumerate}
\end{theorem}
\begin{proof}
\begin{enumerate}
    \item We begin by noting $\alpha(F,x)\leq \alpha(F,I)$ for any point $x\in I$. Thus, $\alpha(F,I)<\frac{13-3\sqrt{17}}{4}$ implies that $\alpha(F,x)<\frac{13-3\sqrt{17}}{4}$, and hence, any point $x\in I$ is an approximate solution to $F$ by Theorem \ref{thm:alpha-1}. 
As $\alpha(F,I)<\infty$, we have the continuity of the $k$-th Newton iteration $N^k_F(x)$ over $I$ for all $k$. Hence, all points in $I$ must converge to the same associated solution $x^\star$. Since $x^\star$ is contained in $B(x,2\beta(F,I))$ for each point $x\in I$, we have that $x^\star$ is contained in $\bigcap\limits_{x\in I}B(x,2\beta(F,I))$. 
    %Lastly, if $2\beta(F,I)<\frac{d}{2}$, then $\bigcap\limits_{x\in I}B(x,2\beta(F,I))$ is an empty set because $B(x_1,2\beta(F,I))$ and $B(x_2,2\beta(F,I))$ do not intersect when $x_1$ and $x_2$ are opposite vertices of the interval vector $I$. Therefore, we have $\beta(F,I)\geq\frac{d}{4}$.
\item    Since $\alpha(F,I)<0.03$, for any point $x\in I$, all points in $B(x,\frac{1}{20\gamma(F,I)})$ are approximate solutions. By the first part of this theorem, all points in $I$ have the same associated solution to $F$ so that the result  is proved.\hfill\squareforqed
\end{enumerate}
\end{proof}

\iffalse
We provide Algorithm \ref{algo:certify_solution} for certifying a candidate solution based on the theorem. Note that $False$ returned from the algorithm does not mean that the candidate solution is not an approximate solution. It may be, but just cannot be determined by the alpha theory computation.

\algrenewcommand\algorithmicrequire{\textbf{Input}:}
\algrenewcommand\algorithmicensure{\textbf{Output}:}
\begin{algorithm}[ht]
	\caption{Certifying a solution to a system of equations}
 \label{algo:certify_solution}
\begin{algorithmic}[1]
\Require An analytic system $F:U\rightarrow\mathbb{C}^n$ defined in an open set $U\subset \mathbb{C}^n$, and an interval vector $I\subset U$ enclosing a candidate solution $x$.   
\Ensure A boolean value that $True$ if $x$ is an approximate solution, $False$ otherwise.
\State {Compute the constant $\alpha(F,I)$.}
\If {$\alpha(F,I)<\frac{13-3\sqrt{17}}{4}$}
\State {Return $True$.}
\Else
\State {Return $False$.}
\EndIf
\end{algorithmic}
\end{algorithm}
\fi
Since the alpha theory over regions certifies all points in a certain region at once, the process of certifying distinct solutions is more relaxed than the known method in \cite{hauenstein2012algorithm}. We introduce it as the following corollary:

\begin{corollary}\label{cor:alpha-distinct}
    Let $F:U\rightarrow\mathbb{C}^n$ be an analytic system, and $I_1$ and $I_2$ be interval vectors in $U$ with
           $\alpha(F,I_1)<\frac{13-3\sqrt{17}}{4}$ and $\alpha(F,I_2)<\frac{13-3\sqrt{17}}{4}$.
    Then, if $I_1\cap I_2\ne\emptyset$, then $I_1$ and $I_2$ have the same associated solution to $F$. On the other hand, if $\text{dist}(I_1,I_2)>2\beta(F,I_1)+2\beta(F,I_2)$, then $I_1$ and $I_2$ have different associated solutions to $F$. Here, $\text{dist}(I_1,I_2)=\min\{\|x_1-x_2\|\mid x_1\in I_1, x_2\in I_2\}$.
\end{corollary}
\begin{proof}
    The first part is clear by applying Theorem 2(\ref{thm:alpha-regions1}) on $I_1$ and $I_2$. The second part follows from the fact that the associated solution of $I_i$ is contained in $B(x_i,2\beta(F,I_i))$ for any point $x_i\in I_i$. \hfill\squareforqed
\end{proof}

Note that the corollary allows a larger alpha value than that of Theorem 1(\ref{thm:distinct}). Hence, it can be used for certifying distinct solutions with a more relaxed condition. %Also, because Theorem 2(\ref{thm:distinct-regions}) is still valid, if $\alpha(F,I)<0.03$, we may obtain a region having the same associated solution which is wider than that from Theorem \ref{thm:alpha-regions}. 

Finally, for the case with a polynomial system, we state the interval version of Proposition \ref{prop:higher-bound}. For a polynomial system $F:\mathbb{C}^n\rightarrow\mathbb{C}^n$ and an interval vector $I$, we define
    $\mu(F,I):=\max_{x\in I}\mu(F,x)$. 
    Then, the proposition below provides an upper bound for $\gamma(F,I)$.
\begin{proposition}\label{prop:higher-bound-interval}
    Consider a polynomial system $F=\{f_1,\dots, f_n\}$, and an interval vector $I=(I_1,\dots, I_n)$ in $\mathbb{C}^n$. 
    Then,
$\gamma(F,I)\leq\frac{\mu(F,I)D^\frac{3}{2}}{2\llfloor(1,I)\rrfloor}$
    where $(1,I)=([1,1]+i[0,0],I_1,\dots, I_n)$ is the interval vector in $\mathbb{C}^{n+1}$.% Remind that $D=\max\limits_i\{d_i\}$ where $d_i$ is the degree of $f_i$.
\end{proposition}
\begin{proof}
Consider any point in $x\in I$. Then,
    \begin{equation*}
        \gamma(F,x)  \leq
        \frac{\mu(F,x)D^\frac{3}{2}}{2\|(1,x)\|}
        \leq
        \frac{\mu(F,I)D^\frac{3}{2}}{2\| (1,x)\|}
        \leq \frac{\mu(F,I)D^\frac{3}{2}}{2\llfloor (1,I)\rrfloor}.
    \end{equation*}
    Since $x$ is an arbitrary point, the result follows.\hfill\squareforqed
\end{proof}

In the actual implementation of alpha theory over regions, we use interval closures $\square\beta(F,I)$ and $\square\mu(F,I)$ instead of $\beta(F,I)$ and $\mu(F,I)$. Computing these interval closures requires the computation of the interval matrix inverse due to $\square JF(I)^{-1}$.

%We discuss the way to compute $\beta(F,I)$ and $\mu(F,I)$ for actual implementation. As both values are defined by the maximum values of $\beta(F,x)$ and $\mu(F,x)$ for all $x\in I$ respectively, we use the upper bounds of them obtained by interval closures. Specifically speaking, for a polynomial system $F=[f_1,\dots, f_n]^\top$ considered as a vector and an interval box $I=(I_1,\dots, I_n)$, we define the interval closure of $F$ over $I$ as an interval vector $\square F(I)=[\square f_1(I),\dots, \square f_n(I)]^\top$. Then, we use $\llceil\square JF(I)^{-1}\square F(I)\rrceil$ as an upper bound for $\beta(F,I)$. Here, $\square JF(I)^{-1}$ is the inverse of the interval matrix $\square JF(I)$. %The computation of the inverse of an interval matrix is discussed in Section \ref{sec:inverse-interval-matrix}. Likewise, we define $\Delta_F(I)$ as an diagonal interval matrix with $$\Delta_F(I)_{ii}:=\sqrt{d_i}\left(1+\sum_{i=1}^n I_i^2\right)^\frac{d_i-1}{2},$$ and use the value $\max\{1,\|F\|\llceil\square JF(I)^{-1}\Delta_F(I)\rrceil\}$ instead of $\mu(F,I)$.

The expected usage of the alpha theory over regions is to replace the usual alpha theory. For a given numerical solution, we apply the alpha theory on an interval vector containing this numerical solution with preferably a small radius. %Executing the alpha theory using interval arithmetic may expedite calculations compared to applying the alpha theory at a specific point with exact arithmetic. In practical scenarios, using an input interval with a small radius is preferred to mitigate the conservative nature of interval arithmetic. %Refer to Section \ref{subsec:experiment-radius} for experimental results demonstrate the variations in output corresponding to changes in the radius of the input interval.

\section{Implementation details}\label{sec:implementation-details}

This section points out remarks when implementing the alpha theory over regions. The first part is about the use of the interval arithmetic library MPFI \cite{revol2005motivations} with arbitrary precision. Secondly, we discuss a special type of interval arithmetic for computing a tight inverse interval matrix using LU decomposition.

\subsection{MPFI for arbitrary precision interval arithmetic}

One of the typical ways to get input for the alpha theory over regions is by constructing an interval vector with a certain radius from a given candidate solution of a system. To execute reliable computation, one may require high precision for interval arithmetic. MPFI  is a library written in C for arbitrary precision interval arithmetic using MPFR \cite{10.1145/1236463.1236468}. The purpose of using MPFI is to achieve guaranteed computation results without losing accuracy from the rounding error. The comparison of alpha values according to the change of precision is presented in Section \ref{sec:precision}. One possible drawback of using high precision is, however, that as the precision used in the computation increases, the speed of the calculation may decrease. The comparison of elapsed time between machine precision and MPFI is presented in Section \ref{sec:comparison}.

\subsection{Inverse interval matrix computation via LU decomposition}\label{sec:inverse-interval-matrix}

For computing the inverse interval matrix, LU decomposition may be considered for efficiency compared to other methods (e.g.\ cofactor expansion). We desire a tight inverse interval matrix since an unnecessarily large inverse makes the alpha value greater than what it could be. To achieve this, we introduce interval arithmetic in a special type.

For an interval $I$ in $\mathbb{C}$, we define its \emph{dual interval} that is denoted by $I^\star$. This dual interval $I^\star$ has the same endpoints of $I$ with the following arithmetic:
\begin{align*}
    I+(-I^\star)&=I-I^\star = [0,0],\\
    I\times \frac{1}{I^\star}&= \frac{I}{I^\star}=[1,1] \quad\text{if }0\not\in I.
\end{align*}
Also, we define $(I^\star)^\star=I$ to make the addition and multiplication commutative. Including the dual intervals introduces \emph{Kaucher arithmetic} with a broader collection of intervals than that of usual interval arithmetic. 
It has a more algebraic structure than the usual interval arithmetic while it executes the conservative computation. More specifically, the set of intervals with dual intervals is a group in addition, and the set of intervals not containing zero with dual intervals is a group in multiplication. The more general version of this interval arithmetic is introduced in \cite{goldsztejn2006generalized}.
The interval arithmetic with dual intervals is used for partial pivoting for LU decomposition. For other subtractions and divisions that occur except for pivoting, the usual interval arithmetic is used. 

We briefly elaborate on how to compute the inverse interval matrix. For an $n\times n$ interval matrix $M$ and an $n$-dimensional interval vector $B$, consider an interval linear system $M\mathbf{x}=B$.
Solving this system using the
interval LU decomposition of $M$, we have an interval vector $\mathbf{x}$ satisfying only $B\subseteq M\mathbf{x}$ in general (See \cite[Section 4.5]{neumaier1990interval}, for example).
Using this fact, iterative solving of interval linear systems returns an inverse interval matrix $M^{-1}$ (that is, it returns a set of interval matrices containing $\{N^{-1}\mid N\in M\}$). In particular, using interval arithmetic with dual intervals may return a tighter $M^{-1}$ than the usual interval arithmetic.
%The following example shows how LU decomposition is performed with interval arithmetic with dual intervals.

\section{Experiments}\label{sec:experiments}

This section provides computational and experimental results as a proof of concept for the alpha theory over regions. The implementation is in \texttt{C++} into two versions, one with double machine precision (in a correct rounding manner for reliable computation), and the other with arbitrary precision using MPFI \cite{revol2005motivations}. It computes alpha, beta, and gamma values from a given square polynomial system and a point. 
%Refining the given solution, sorting distinct solutions, or determining the realness of solutions are not provided as they are not the main focus.
All computations in this section are executed in a Macbook M2 pro 3.5 GHz, 16 GB RAM. The code and examples are available at

\begin{center}
\url{https://github.com/klee669/alphaTheoryOverRegions}    
\end{center} 

\subsection{Alpha values according to the radius of the input interval}\label{subsec:experiment-radius}
We analyze the impact of the radius of the interval vector on alpha values. It is expected that as the size of the interval increases, the alpha value will also increase. To check this, we consider the cyclic system with $6$ variables;  $f_l=\sum_{j=1}^6\prod_{k=j}^{j+l}x_j$ for $l=1,\dots,5$ and $f_6 = \prod_{j=1}^6x_j-1$
with a numerical solution $a=(a_1,-a_1,a_2,a_3,-a_3,-a_2)$ where $a_1=.782290+.622915i, a_2=.866025-.5i$ and $a_3=.148315+.988940i$ whose distance from the nearest exact solution is $8.261221e-7$. Defining the interval vector centered at $a$, we compute three constant values using our implementation with double precision while we change the radius of the interval (see Table \ref{table:radius}). For reference, values computed by \texttt{alphaCertified} \cite{hauenstein2011alphacertified} with exact arithmetic are also recorded. 

\begin{table}
\centering
\begin{tabular}{c|c|c|c}
    radius & alpha & beta & gamma \\
    \hhline{=|=|=|=}
    $1e-5$ & $1.15611e+1$ & $4.03387e-3$ & $2.86602e+3$ \\
    $1e-7$ & $1.66585e-2$ & $1.26115e-5$ & $1.32089e+3$ \\
    $1e-10$ & $1.62002e-5$ & $1.23899e-8$ & $1.30754e+3$ \\
    $1e-15$ & $5.75515e-11$ & $4.40158e-14$ & $1.30752e+3$ \\
    $1e-20$ & $1.63550e-11$ & $1.25084e-14$ & $1.30752e+3$ \\
    $1e-30$ & $1.63550e-11$ & $1.25084e-14$ & $1.30752e+3$ \\
    \texttt{aC} & $1.53040e-13$ & $1.17046e-16$ & $1.30752e+3$ \\
\end{tabular}
\caption{The values of alpha, beta, and gamma constants for the cyclic-$6$ system according to the change of the radius of the input interval. All computations were conducted with double precision. The last row shows the values obtained by the software \texttt{alphaCertified} with exact arithmetic.}
\label{table:radius}
\end{table}

From the result, each constant gets larger as the radius increases. The improvement in alpha and beta are more noticeable than that of gamma because the beta value is affected by the value of $F(x)$ which is more sensitive to the change of the size of the input interval than the bound for gamma value given in Proposition \ref{prop:higher-bound-interval}. Once the radius gets small enough, due to the conservative computation from interval arithmetic, it shows only a slight improvement on three constant values. %On the other hand, constant values computed by our implementation are larger than those obtained by \texttt{alphaCertified}. 

\subsection{Alpha values according to precision}\label{sec:precision}
In this section, we explore how precision affects the values.
%for alpha theory using the implementation employing MPFI. 
%Even though using double precision is proper for fast computation, higher precision is still needed when the problem requires more refined input. Especially, when the system has more polynomials with higher degrees, the rounding error from double precision may increase. 
We consider the same cyclic-$6$ system, and the interval box with the radius $1e-20$ centered at the solution $a$ used in Section \ref{subsec:experiment-radius}. The comparison of alpha, beta, and gamma values according to the change of precision is given in Table \ref{table:precision}. 

\begin{table}
\centering
\begin{tabular}{c|c|c|c}
    precision & alpha & beta & gamma \\
    \hhline{=|=|=|=}
    $16$ & $4.44465e-2$ & $3.39918e-5$ & $1.30757e+3$ \\
    $32$ & $4.42851e-7$ & $3.38696e-10$ & $1.30752e+3$ \\
    $64$ & $2.03482e-13$ & $1.55624e-16$ & $1.30752e+3$ \\
    $128$ & $2.02057e-13$ & $1.54534e-16$ & $1.30752e+3$ \\
    $256$ & $2.02057e-13$ & $1.54534e-16$ & $1.30752e+3$ \\
    double & $1.63550e-11$ & $1.25084e-14$ & $1.30752e+3$
\end{tabular}
\caption{The values of alpha, beta, and gamma constants for the cyclic-$6$ system according to the change of precision.}
\label{table:precision}
\end{table}

The result shows that the smaller alpha and beta values are returned as the larger precision is used. The value of gamma does not improve much since the beta value is affected by the value of $F(x)$ which is more sensitive to the change of precision. The changes in all three values become insignificant when the precision higher than $128$ is used.

\subsection{Time comparison with \texttt{alphaCertified}}\label{sec:comparison}

%As computation via interval arithmetic may avoid costly exact arithmetic computation, 
We provide a time comparison with the software \texttt{alphaCertified}. 
We experiment with the Fano problem studied in \cite{yahl2023computing}. The Fano problem of type $(n,k,d)$ where $d=(d_1,\dots, d_l)$ is the problem of finding $n$-dimensional planes lying in a complete intersection of $l$ hypersurfaces $f_1,\dots, f_l$ in $\mathbb{P}^k$ with degrees $\deg f_1= d_1,\dots, \deg f_l=d_l$. Fano problems can be described as problems of solving a square polynomial system. For example, Fano problems of $(1,5,(2,4))$ is related to a square polynomial system of $8$ variables with $1280$ solutions (up to multiplicity), and $(1,8,(2,2,2,4))$ is related to a square polynomial systems of $14$ variables with $47104$ solutions (up to multiplicity). We find numerical solutions of these two systems with \texttt{Macaulay2} expressed in floating point arithmetic, and certify them using our implementation and \texttt{alphaCertified} by varying the number of candidate solutions. For our implementation, we compute alpha, beta, and gamma values for each candidate solution using both double precision and MPFI with $256$ precision. For \texttt{alphaCertified}, calculations are performed both using exact arithmetic and floating-point arithmetic for comparison even though floating-point arithmetic only provides soft verification. The result is recorded in Table \ref{table:Fano}.
%In the supplementary material \cite{YahlData} of \cite{yahl2023computing}, the author pointed out that \texttt{alphaCertified} took a significantly long computation time for certifying the computation of the Galois group of Fano problems. 

\begin{table}
\centering
    \begin{minipage}{.49\linewidth}
\begin{flushleft}
$(1,5,(2,4))$, a square system with $8$ variables.
\vspace{-.5pc}
\end{flushleft}
\begin{tabular}{c|c|c|c|c}
    \#sols & double & $256$ prec. &\texttt{aC} exact & \texttt{aC} float \\ 
    \hhline{=|=|=|=|=}
    $20$ & $.06$ & $.67$ & $92.92$ & $.16$\\
    $50$ & $.11$ & $1.59$& $242.00$ & $.27$\\
    $200$ & $.31$ & $6.13$ &  $1067.02$ & $.84$\\
    $1000$ & $1.39$ & $32.83$ &  $7649.28$ & $3.92$
\end{tabular}
%\bigskip
\end{minipage}
    \begin{minipage}{.49\linewidth}
\begin{flushleft}
$(1,8,(2,2,2,4))$, a square system with $14$ variables.
\vspace{-.5pc}
\end{flushleft}
\begin{tabular}{c|c|c|c|c}
    \#sols & double & $256$ prec.&\texttt{aC} exact & \texttt{aC} float \\
    \hhline{=|=|=|=|=}
    $20$ & $.41$ & $7.78$ &  $8743.73$ & $1.51$\\
    $50$ & $.86$ & $18.26$ &  $22500.32$ & $2.28$\\
    $200$ & $3.07$ & $72.10$ &  $89455.63$ & $6.73$\\
    $1000$ & $14.73$ & $400.51$&  $-$ & $27.61$
\end{tabular}
\end{minipage}%

\caption{Elapsed time in seconds for certifying solutions for Fano problems of using the implementation of alpha theory over regions, and the software \texttt{alphaCertified}. The symbol $-$ means that the computation does not terminate within $2$ days.}
\label{table:Fano}
\end{table}

The result shows that the alpha theory over regions shows less elapsed time on computation than \texttt{alphaCertified} with exact arithmetic. The implementation with $256$ precision may take more time than \texttt{alphaCertified} with floating point arithmetic, but it returns more reliable results than computation with floating point arithmetic. The implementation with double precision takes less time than that with MPFI. 

Note that comparing the elapsed time of two software might not be fair since \texttt{alphaCertified} performs further analysis to classify distinct solutions. Nonetheless, the result shows the potential of the alpha theory over regions as it produces reliable results in a significantly shorter time than the computation with exact arithmetic.

\begin{credits}
\subsubsection{\ackname} We would like to thank Michael Burr and Thomas Yahl for helpful discussions. We also thank the anonymous referees for their constructive comments.

\subsubsection{\discintname}
The authors have no competing interests to declare that are
relevant to the content of this article.
\end{credits}
%
% ---- Bibliography ----
%
% BibTeX users should specify bibliography style 'splncs04'.
% References will then be sorted and formatted in the correct style.
%
 \bibliographystyle{splncs04}
 \bibliography{../sample-base}
\end{document}